\documentclass[11pt]{article}
\newcommand{\PAPER}[1]{#1}
\newcommand{\LIPICSv}[1]{}

\PAPER{
\usepackage{fullpage,times,amsmath,amssymb,amsthm,hyperref,graphicx}

\newtheorem{theorem}{Theorem}[section]
\newtheorem{lemma}[theorem]{Lemma}

\title{Triangulating a Polygon with Holes in\\ Optimal (Deterministic) Time}
\author{Timothy M. Chan\thanks{
  Siebel School of Computing and Data Science, University of Illinois at Urbana-Champaign
  (tmc@illinois.edu)}}

}

\LIPICSv{

\usepackage[utf8]{inputenc}

\usepackage{microtype,hyperref,graphicx}
   \hypersetup{%
      breaklinks,%
      colorlinks=true,%
      urlcolor=[rgb]{0.25,0.0,0.0},%
      linkcolor=[rgb]{0.5,0.0,0.0},%
      citecolor=[rgb]{0,0.2,0.445},%
      filecolor=[rgb]{0,0,0.4},
      anchorcolor=[rgb]={0.0,0.1,0.2}%
   }

\title{Triangulating a Polygon with Holes in Optimal (Deterministic) Time}

\author{Timothy M. Chan}{Siebel School of Computing and Data Science, University of Illinois at Urbana-Champaign, USA}{tmc@illinois.edu}{https://orcid.org/0000-0002-8093-0675}{}

\titlerunning{}
\authorrunning{T.\,M. Chan}

\Copyright{Timothy M. Chan}



\renewcommand{\paragraph}[1]{\subparagraph*{#1}}

}

\newcommand{\eps}{\varepsilon}
\newcommand{\TD}{\textit{TD}}

\newcommand{\DEG}{\textrm{deg}}
\newcommand{\polylog}{\mathop{\rm polylog}}

\begin{document}

\maketitle

\begin{abstract}
We consider the problem of triangulating a polygon with $n$ vertices and $h$ holes, or relatedly the problem of computing the trapezoidal decomposition of a collection of $h$ disjoint simple polygonal chains with $n$ vertices total. Clarkson, Cole, and Tarjan (1992) and Seidel (1991) gave randomized algorithms running in $O(n\log^*n + h\log h)$ time, while Bar-Yehuda and Chazelle (1994) described deterministic algorithms running in $O(n+h\log^{1+\varepsilon}h)$ or $O((n+h\log h)\log\log h)$ time, for an arbitrarily small positive constant $\varepsilon$. No improvements have been reported since. We describe a new $O(n + h\log h)$-time algorithm, which is optimal and deterministic.

More generally, when the given polygonal chains are not necessarily simple and may intersect each other, we show how to compute their trapezoidal decomposition (and in particular, compute all intersections) in optimal $O(n + h\log h)$ deterministic time when the number of intersections is at most $n^{1-\varepsilon}$.

To obtain these results, Chazelle's linear-time algorithm for triangulating a simple polygon is used as a black box.
\end{abstract}

\section{Introduction}

Triangulating a simple polygon is one of the most basic problems in computational geometry, with numerous applications.
Chazelle's  deterministic linear-time algorithm~\cite{Chazelle91} represents one of the landmark results in the area,
improving previous deterministic $O(n\log\log n)$-time algorithms by Tarjan and Van Wyck~\cite{TarjanW88} and Kirkpatrick, Klawe, and Tarjan~\cite{KirkpatrickKT92}
and previous randomized $O(n\log^*n)$-time algorithms by Clarkson, Tarjan, and Van Wyck~\cite{ClarksonTW89} and Seidel~\cite{Seidel91}.  Here, $n$ denotes the number of vertices in the input polygon and $\log^*$ is the (very slow-growing) iterated logarithm function.

In this paper, we study an extension of the problem: triangulating a polygon with $h$ holes (where the holes themselves are disjoint simple polygons enclosed in an outer simple polygon).  See Figure~\ref{fig:polygon} (left).
The problem can be equivalently stated in a number of different ways:

\begin{enumerate}
\item Computing the \emph{trapezoidal decomposition} of a polygon with $h$ holes.
The decomposition (also called the ``trapezoidization'', or ``vertical decomposition'',
or ``vertical visibility map'') is obtained by adding vertical line segments from each vertex to
the edge immediately above/below the vertex.
Linear-time reductions from trapezoidal decompositions to triangulations and vice versa are known~\cite{FournierM84}.

\LIPICSv{\smallskip}
\item Computing the trapezoidal decomposition of a collection of $h$ disjoint simple polygonal chains.
See Figure~\ref{fig:polygon} (right).
(We can enclose the chains with a bounding box and turn each chain into an infinitesimally thin hole,
while increasing $n$ and $h$ by only a constant factor.)

\LIPICSv{\smallskip}
\item Computing the trapezoidal decomposition of a plane straight-line graph (a graph drawn on the plane with straight-line edges without crossings, where edges are ordered around each vertex) with $h$ connected components.
(We can make tiny cuts to eliminate cycles and then take an Euler tour to turn each component into
a chain, while increasing $n$ and $h$ by only a constant factor.)
\end{enumerate}

\begin{figure}
\centering
\PAPER{
\includegraphics[page=1,scale=0.5]{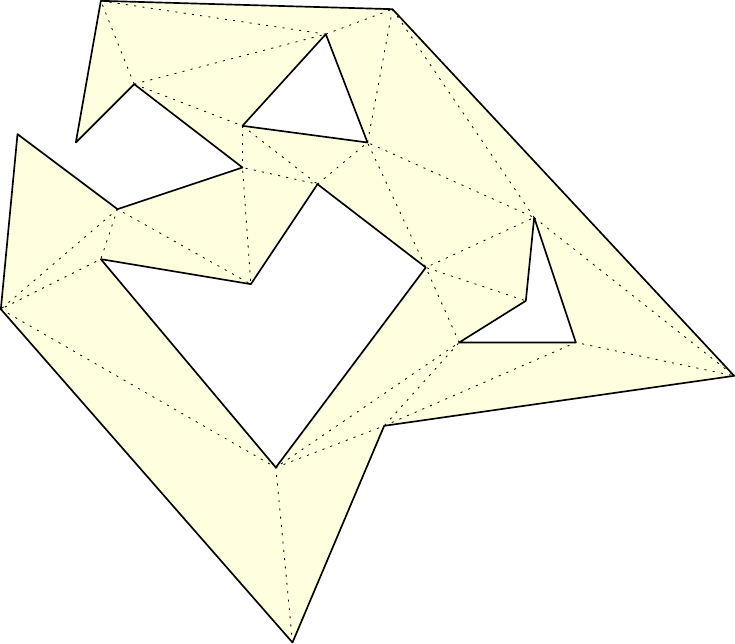}\hspace{5em}%
\includegraphics[page=2,scale=0.5]{fig_polygon.pdf}
}
\LIPICSv{
\includegraphics[page=1,scale=0.43]{fig_polygon.pdf}\hspace{5em}%
\includegraphics[page=2,scale=0.43]{fig_polygon.pdf}
}
\caption{(Left) A triangulation of a polygon with 3 holes.  (Right) The trapezoidal decomposition of 5 disjoint simple polygonal chains.}\label{fig:polygon}
\end{figure}


Previous randomized $O(n\log^*n)$-time algorithms for simple polygons can be adapted
to solve the problem in $O(n\log^*n + h\log h)$ expected time~\cite{ClarksonCT92,Seidel91}.
In 1994, Bar-Yehuda and Chazelle~\cite{Bar-YehudaC94} gave deterministic algorithms
running in $O(\min\{n+h\log^{1+\eps}h,\,(n+h\log h)\log\log h\})$ time for an arbitrarily small constant $\eps>0$.  These results came close to the known lower bound of $\Omega(n + h\log h)$,%
\footnote{
An $\Omega(h\log h)$ lower bound in the algebraic decision tree model can be shown by a simple reduction from sorting: to sort $h$ numbers $x_1,\ldots,x_h\in (-1,1)$, create a polygon that has a triangular outer boundary with vertices $(-1,1),(1,1),(0,-1)$, and a sufficiently small hole near $(x_i,x_i^2)$ for each $i\in\{1,\ldots,h\}$.
}
but there is still
an extra factor of $\log^*n$ for randomized algorithms, and $\log^\eps h$ or $\log\log h$ for deterministic.  No improvement has been reported since, deterministic or randomized.

Meanwhile, a number of papers (e.g.,~\cite{ChenW15,KapoorM00,MitchellPSW19,Wang23}) have given algorithms
for various problems related to shortest paths and visibility in polygonal domains (polygonal environments with multiple  disjoint polygonal obstacles) whose time complexities are dominated by that of the triangulation problem and would be optimal if the triangulation complexity is $O(n + h\log h)$.  However, with the current gap of knowledge, these papers needed to state their results
more awkwardly, either by assuming that a triangulation of the free space of the polygonal domain (a polygon with holes) is given, or by adding an extra term $T$ to the time bound where $T$ represents the unknown complexity of the triangulation problem.  (A similar story occurred for various computational problems about simple polygons in the days before Chazelle's linear-time algorithm.)


\paragraph{Main result.}
We give the first optimal algorithm for triangulating a polygon with holes in $O(n + h\log h)$ time.
The algorithm is deterministic.

Many applications follow~\cite{ChenW15,KapoorM00,MitchellPSW19,Wang23}. (For example,
combining with Wang's result~\cite{Wang23}
yields an optimal $O(n+h\log h)$-time algorithm for finding the Euclidean shortest path between two given points
amidst $h$ disjoint polygonal obstacles with $n$ vertices total.)

\paragraph{Previous approaches.}
Known approaches, as well as our new algorithm, work with the reformulation of the problem in terms of the trapezoidal decomposition of $h$ disjoint polygonal chains.

One could try to directly modify Chazelle's linear-time algorithm for $h=1$~\cite{Chazelle91}, but this appears difficult (considering how notoriously complicated his algorithm was).
Bar-Yehuda and Chazelle's approach was to directly reduce the general $h$ case to the $h=1$ case:
join the multiple polygonal chains into one connected component by adding vertical line segments, and
then apply Chazelle's algorithm as a black box.
To do so, the first observation is that polygonal chains intersecting a common vertical line can
be joined together this way by \emph{Jordan sorting} (the problem of sorting intersections of polygonal chains along a line), which is solvable by known techniques~\cite{HoffmanMRT86}.  Roughly,
their idea was then to build an \emph{interval tree}~\cite{BergCKO08,PreparataS85} for the $x$-spans of the chains, and solve the
problem for each level of the interval tree (since each level consists of vertically separated instances in which
chains intersect a common vertical line).
The remaining task is merging the results of these $O(\log h)$ levels.
Bar-Yehuda and Chazelle proposed
a binary divide-and-conquer over the $O(\log h)$ levels, yielding $O((n+h\log h)\log\log h)$ time overall,
and alternatively a divide-and-conquer with a larger fan-out, yielding $O(n + h\log^{1+\eps}h)$.
This approach seems hard to improve further.

A different approach is randomized divide-and-conquer, as done by Clarkson et al.\ and Seidel \cite{ClarksonTW89,Seidel91}.  The trapezoidal decomposition of a random sample of $n/\log n$ edges
allows us to divide the problem into subproblems of roughly logarithmic size.  Assigning input elements to subproblems requires some form of ``point location'' and would normally cost $O(\log n)$ per element, but can be done faster knowing that the input elements comes from $h$ chains.  Since the recursion has $O(\log^*n)$ levels, the running time inherently increases by a $\log^*n$ factor, unfortunately.  Even for $h=1$ (the case of a simple polygon), Clarkson et al.\ and Seidel's methods do not achieve linear time (Amato, Goodrich, and Ramos~\cite{AmatoGR01} described a randomized linear-time algorithm which is simpler than Chazelle's,
but still relies on some of Chazelle's ingredients and is still complicated).
Furthermore, randomization seems essential: such random sampling methods where the sample size is large is generally harder to derandomize.

\paragraph{Our approach.}
The difficulty of improving the random sampling approach turns out to be overstated.
We observe that in combination with another known divide-and-conquer
strategy via \emph{planar graph separators} \cite{LiptonT80} or \emph{$r$-divisions} \cite{Frederickson87,Goodrich95}, the random sampling approach can indeed achieve optimal $O(n + h\log h)$
expected running time.  The reason the $\log^*n$ factor is avoidable is that we can
recurse for a constant number of levels (in fact, just once) and then switch to
Chazelle's algorithm (or more precisely Bar-Yehuda and Chazelle's algorithm, which calls
Chazelle's) in the base case when the input size is sufficiently small---the $h\log h$ term is big
enough to absorb the cost of the switch.  The idea of combining sampling with $r$-divisions and invoking Chazelle's in the base case is not original and can be viewed as
a variant of an existing randomized algorithm by
Eppstein, Goodrich, and Strash~\cite{EppsteinGS10} for a related problem
(see below).  Still, it is interesting that this implication has not been noted earlier, and we include this optimal randomized solution in Appendix~\ref{app:rand}.

Our main contribution is in obtaining an optimal deterministic algorithm.
To do so, we combine ideas from both the random sampling approach and Bar-Yehuda and Chazelle's deterministic approach, even though the two approaches appear different and incompatible.
To derandomize the sampling approach, we need to compute a subset of the input edges that ``behave'' like a random sample,
known as an \emph{$\eps$-net} in the literature~\cite{HausslerW87,Matousek00} (with $\eps$ being $\frac{\polylog n}{n}$ in our case).
Constructing a net deterministically appears challenging, especially since we can only afford linear time.
We observe a simple way to construct a net in linear time when the trapezoidal decomposition is given, 
by using planar graph separators or $r$-divisions (the net size is suboptimal, but is good enough for our purposes).  

However, there is an obvious circularity issue: we are not given the trapezoidal decomposition, since that is what we want to compute!  Our key idea is to go back to Bar-Yehuda and Chazelle's interval-tree approach, which gives a partition of the input into $O(\log n)$ subsets where we already know how to compute the trapezoidal decomposition of each subset.  Although merging their trapezoidal decompositions is not easy, merging their nets is:
the union of the nets of the subsets is a net for the whole input (the net size is worsen by an $O(\log n)$ factor, but is still good enough for our purposes).   By building the trapezoidal decomposition of the net
in combination with an $r$-division, we can then divide the problem into subproblems of polylogarithmic size.  The detailed description of the whole algorithm is given in Section~\ref{sec:disjoint}, and is just about 2 pages long.

\paragraph{When there are intersections.}\label{sec:intro:intersect}
In the last part of the paper, we address an extension of the problem: computing the trapezoidal decomposition of
polygonal chains that may not be simple, i.e., may self-intersect, and may intersect each other.
The decomposition is defined in the same way, except that we also add vertical segments from
each intersection point to the edge immediately above/below.
The crux of the problem lies in the computation of the intersections.  (Once all the intersection points are known, we could actually reduce to the non-intersecting case to obtain the final trapezoidal decomposition.)

The problem of computing all self-intersections of one chain (the $h=1$ case) is already challenging.
Chazelle~\cite{Chazelle91} posed the (still open) question of whether there is an $O(n+X)$-time algorithm, where $X$
denotes the number of intersections.
Known line-segment intersection algorithms~\cite{ClarksonS89,Mulmuley90,ChazelleE92,Balaban95} can solve the problem in $O(n\log n + X)$ time, which is optimal when $X\gg n\log n$ but not when $X$ is sublinear; they do not exploit the fact that 
the input line segments form a connected chain.
Clarkson, Cole, Tarjan~\cite{ClarksonCT92} showed that the random sampling approach yields
an $O(n\log^* n + X)$ expected time bound for this problem.
On the other hand,
Eppstein, Goodrich, and Strash~\cite{EppsteinGS10} used random sampling in combination with $r$-divisions to
obtain an alternative $O(n + X\log^{(c)}n)$ expected time bound for an arbitrarily large constant~$c$, where
$\log^{(c)}$ denotes the logarithm function iterated $c$ times.
Thus, their algorithm runs in optimal linear expected time when $X \ll n/\log^{(c)}n$.
We are not aware of any \emph{deterministic} result better than $O(n\log n+X)$, however.

We present a deterministic algorithm for computing the trapezoidal decomposition of
a non-simple polygonal chain with $X$ intersections in $O(n+Xn^\eps)$ time.
Thus, our algorithm runs in optimal linear time when $X\ll n^{1-\eps}$.
(This case is already sufficient for various applications; for example,
Eppstein and Goodrich~\cite{EppsteinG08} provided empirical evidence that road networks tend to
have $X\approx \sqrt{n}$ intersections.)  

More generally, when there are multiple chains ($h>1$),
Clarkson, Cole, and Tarjan~\cite{ClarksonCT92} obtained an $O(n\log^*n + h\log h +X)$ expected time bound.
Our deterministic algorithm, presented in
Section~\ref{sec:intersect}, takes $O(n + h\log h + Xn^\eps)$ time.
(In Appendix~\ref{app:rand}, we also note an extension of Eppstein, Goodrich, and Strash's randomized algorithm with $O(n + h\log h + X\log^{(c)}n)$ expected time bound.)

\section{Preliminaries}

We begin with some standard definitions.

\paragraph{Trapezoidal decomposition.}
For a set $S$ of $n$ line segments in the plane, the \emph{trapezoidal decomposition}, denoted $\TD(S)$,
is the subdivision of the plane formed by $S$ together with vertical line segments drawn from
each vertex $v$ of the arrangement (i.e., each endpoint and each intersection point) to the segment 
immediately above and below $v$.  The faces in this subdivision are indeed trapezoids (or triangles)
where the left and right sides (if exist) are vertical and the top and bottom sides (if exist) are
parts of the input segments of $S$.
If $S$ is disjoint, $\TD(S)$ has $O(n)$ trapezoids and can be constructed in $O(n\log n)$ time~\cite{BergCKO08,PreparataS85}.  If $S$ has $X$ intersections, $\TD(S)$ has $O(n+X)$ trapezoids and
can be constructed in $O((n+X)\log n)$ time by the classic Bentley--Ottmann sweep~\cite{BentleyO79,PreparataS85}
or in $O(n\log n + X)$ time by the fastest known algorithms~\cite{ClarksonS89,Mulmuley90,ChazelleE92}.

The trapezoidal decomposition of a collection $P$ of (not necessarily simple) polygonal chains,
denoted $\TD(P)$, is defined as the trapezoidal decomposition of its edge set.

\paragraph{$r$-division.}
We use the following weak definition of $r$-divisions~\cite{Frederickson87}:
For a planar graph $G$ with $n$ vertices, an \emph{$r$-division} is a
subset $B$ of $O(n/\sqrt{r})$ vertices such that every connected component in $G-B$ has
size at most $r$.
An $r$-division can be viewed as an extension of planar graph separators, and
can be constructed in $O(n)$ time by known algorithms~\cite{Goodrich95,KleinMS13} for any given $r$. 

In our application, we consider an $r$-division of the dual\footnote{In the dual graph, the vertices are the trapezoids, and two trapezoids are adjacent if their boundaries (vertical sides or parts of input line segments) overlap.}
 of a trapezoidal decomposition.
Here, $B$ is a subset of trapezoids.
Construct the subdivision of the plane formed by the sides of the trapezoids in $B$.
The cells of this subdivision will be referred to as the \emph{regions} of the $r$-division in this paper.
Each region is a polygon possibly with holes, which is the union of at most $r$ trapezoids (a region could
also just be a single trapezoid, for example, when it is from $B$).  The total number of sides in all the regions\footnote{
We use sides instead of edges to avoid confusion between geometric and graph terminologies.
} is $O(n/\sqrt{r})$.
(With a stronger version of $r$-divisions~\cite{Frederickson87,KleinMS13}, it is possible to ensure that there are $O(n/r)$
regions, each having at most $O(\sqrt{r})$ sides and at most $O(1)$ holes, though this will not be necessary for our application.)

\paragraph{Interval tree.}
For a set ${\cal I}$ of $h$ intervals, the \emph{interval tree}~\cite{BergCKO08,PreparataS85}
is a binary tree defined recursively as follows.
If ${\cal I}=\emptyset$, the tree is just a single node.
Otherwise, let $a$, $b$, $m$ be the minimum, maximum, and median endpoint.
At the root, store $[a,b]$, which we call the \emph{range}, and $m$, which we call the \emph{pivot value}.
At the root, also store all intervals of ${\cal I}$ containing $m$.
Recursively build the interval tree for the subset of all intervals of ${\cal I}$ completely left of $m$ (this is the left subtree)
and the interval tree for the subset of all intervals of ${\cal I}$ completely right of $m$ (this is the right subtree).
The resulting tree has $O(\log h)$ height, and can be constructed in $O(h\log h)$ time.  
Note that each interval of ${\cal I}$ is stored in exactly one node.

\section{Computing the trapezoidal decomposition of $h$ disjoint simple polygonal chains}\label{sec:disjoint}

%
%
%
%

We now present our main deterministic algorithm for computing the
 trapezoidal decomposition of $h$ disjoint simple polygon chains (from which 
 an algorithm for triangulating a polygon with $h$ holes would follow).

\begin{theorem}\label{thm:disjoint}
Given a collection $P$ of $h$ disjoint simple polygonal chains with a total of $n$ vertices,
we can compute $\TD(P)$ in $O(n + h\log h)$ deterministic time.
\end{theorem}
\begin{proof}
Let $s$ and $t$ be parameters to be set later.
Our algorithm is described in a series of steps below.
The purpose of Steps 0--2 is to compute a subset $R$ satisfying a certain property.
Afterwards, Steps 3--5 use its trapezoidal decomposition $\TD(R)$ to divide the problem into smaller subproblems and then solve these subproblems.

\begin{enumerate}
\item[\LIPICSv{\bf\textcolor{lipicsGray}} 0.]
For each chain $p\in P$, let $I_p$ be the interval determined by the smallest
and largest $x$-coordinate in~$p$; these can be generated in $O(n)$ time.
Build the interval tree for these $O(h)$ intervals in $O(h\log h)$ time.
For each node $v$ of the tree, let $P_v$ be the list of all chains $p\in P$ whose
intervals $I_p$ are stored at~$v$.  Note that each chain appears in exactly one list.


\LIPICSv{\medskip}
\item
For each node $v$ of the interval tree with pivot value $m_v$:
\LIPICSv{\smallskip}
\begin{enumerate}
\item Sort the intersections of $P_v$ with the vertical line $x=m_v$ (note that
all chains in $P_v$ intersect this line).
This is an instance of the \emph{Jordan sorting} problem for multiple disjoint polygonal chains.
Bar-Yehuda and Chazelle~\cite{Bar-YehudaC94} showed that this can be done in $O(n_v+h_v\log h_v)$ time (by
adapting a known Jordan sorting algorithm for a single polygonal chain~\cite{HoffmanMRT86}), where
$n_v$ and $h_v$ are the number of vertices and chains in $P_v$ respectively.
Note that $\sum_v n_v = n$ and $\sum_v h_v = h$.

\LIPICSv{\smallskip}
\item 
Join all chains in $P_v$ into a single simple, connected component $P'_v$
by adding vertical line segments between all pairs of consecutive intersection points along $x=m_v$, by using the output of step 1(a). (Not all such vertical segments are needed to connect the chains, as shown in Figure~\ref{fig:jordan:sort} (left), but this would not matter.)  Apply Chazelle's algorithm~\cite{Chazelle91} to
compute $\TD(P'_v)$, and thus $\TD(P_v)$, in $O(n_v)$ time.

\LIPICSv{\smallskip}
\item
Compute an $s$-division of the dual of $\TD(P_v)$ in $O(n_v)$ time.
Preprocess the regions of the $s$-division for point location~\cite{PreparataS85} in linear time.
Let $R_v$ be the subset of all $O(n_v/\sqrt{s})$ edges 
of $P_v$ that participate in the region boundaries in this $s$-division (this excludes vertical sides).
Then $R_v$ obeys the following property:\footnote{
This property basically says that $R_v$ is an \emph{$O(s/n_v)$-net}~\cite{HausslerW87,Matousek00} of the edge set of $P_v$, in the set system corresponding to edges intersecting vertical line segments.
The size of our net ($O(n_v/\sqrt{s})$) isn't the smallest possible (which would be $O(n_v/s)$) but is good enough for our purposes.
}
\LIPICSv{\smallskip}
\begin{quote}
($\ast$)\ \
Given any vertical line segment $e$ not intersecting $R_v$, the segment $e$ can intersect at most $O(s)$ edges of $P_v$; furthermore, these edges can be found in $O(\log n + s)$ time.
\end{quote}
\LIPICSv{\smallskip}
This property is easy to see: any such segment $e$ must lie completely inside a region in the
$s$-division, which can be found by point location~\cite{PreparataS85} in $O(\log n)$ time, and there are at most $O(s)$ edges of $P_v$ participating in the region.
\end{enumerate}

\LIPICSv{\medskip}
\item
Let $R$ be the union of all the $R_v$'s, together with $O(h)$ extra line segments of zero length 
at the start vertices of the chains of $P$.  Note that $|R|=O(n/\sqrt{s} + h)$,
and $R$ satisfies the following property:\footnote{
This property basically says that $R$ is an $O((s\log h)/n)$-net of the edge set of $P$ (ignoring the
zero-length segments).
}
\LIPICSv{\smallskip}
\begin{quote}
($\dag$)\ \
Given any vertical line segment $e$ not intersecting $R$, the segment $e$ can intersect at most $O(s\log h)$ edges of $P$; furthermore, these edges can be found in $O((\log n + s)\log h)$ time.
\end{quote}
\LIPICSv{\smallskip}
This follows immediately from property ($\ast$), since the $x$-coordinate of a vertical segment $e$ can lie in the ranges
of only $O(\log h)$ nodes $v$ of the interval tree.

\begin{figure}
\centering
\PAPER{\includegraphics[scale=0.7]{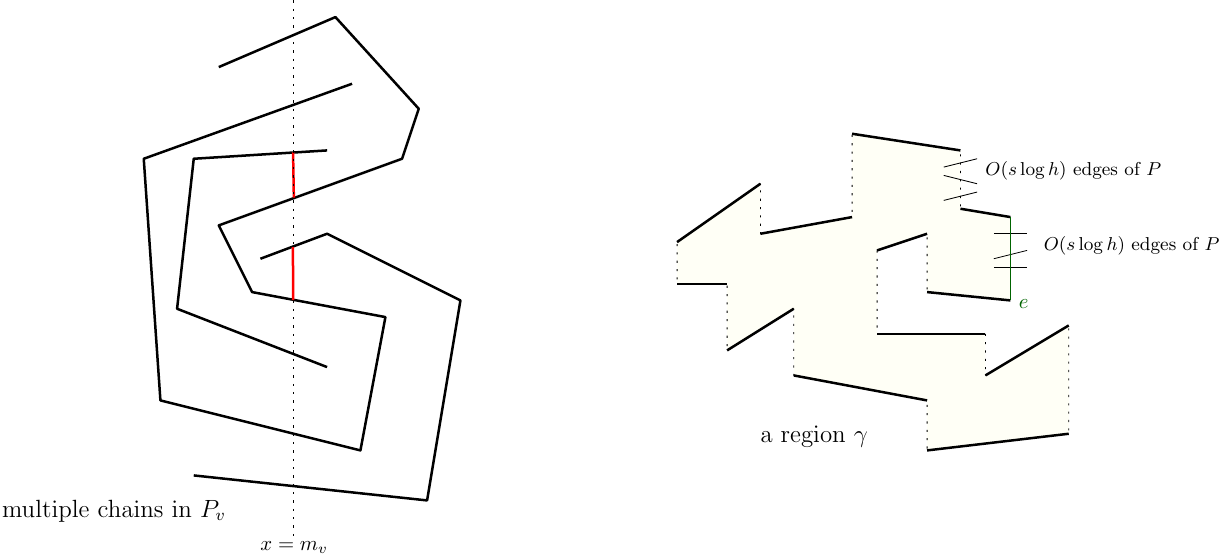}}
\LIPICSv{\includegraphics[scale=0.6]{fig_jordan_sort.pdf}}
\caption{(Left) In Step~1b, we can join multiple chains in $P_v$ intersecting a vertical line $x=m_v$ into a single component, for example, by adding the red vertical segments shown.  (Right) In Step~4, we find the $O(s\log h)$ edges of $P$ intersecting each vertical side $e$ of a region $\gamma$; the non-vertical sides are parts of edges of~$R$.
Each region $\gamma$ is the union of $O(t)$ trapezoids of $\TD(R)$, and the total number of sides over all regions $\gamma$ is $O(|R|/\sqrt{t})$.}\label{fig:jordan:sort}
\end{figure}

\LIPICSv{\medskip}
\item
Compute $\TD(R)$ in $O(|R|\log|R|)$ time, and
compute a $t$-division of the dual of $\TD(R)$ in $O(|R|)$ time.

\LIPICSv{\medskip}
\item
For each region $\gamma$ in the $t$-division and
each vertical side $e$ of the boundary $\partial\gamma$,
compute all intersections of $e$ with $P$.
Since the interior of $e$ does not intersect $R$, we know that $e$ has at most $O(s\log h)$ intersections with $P$  by property ($\dag$), and
they can be found in $O((\log n + s)\log h)$ time.  (See Figure~\ref{fig:jordan:sort} (right).)
By summing over all $O(|R|/\sqrt{t})$ choices of $e$, there are at most $O((|R|/\sqrt{t})\cdot s\log h)$ such intersections in total
and they can be found in $O((|R|/\sqrt{t})\cdot (\log n+s)\log h)$ time.

\LIPICSv{\smallskip}
Cut the chains of $P$ at all these intersection points to obtain a new collection $P'$ of polygonal chains.
In order to cut,
we may need to sort the intersections along each edge of $P$, but this costs at most a logarithmic factor
per intersection, i.e., $O((|R|/\sqrt{t})\cdot s\log h\cdot \log n)$ in total.
The number of vertices in $P'$ is $n'\le n + O((|R|/\sqrt{t})\cdot s\log h)$,
and the number of chains in $P'$ is $h'\le h + O((|R|/\sqrt{t})\cdot s\log h)$.

\LIPICSv{\medskip}
\item 
For each region $\gamma$ in the $t$-division,
let $P'_\gamma$ be the chains of $P'$ inside $\gamma$,
and compute $\TD(P'_\gamma)$ by applying Bar-Yehuda and Chazelle's previous algorithms~\cite{Bar-YehudaC94}
in $O(n'_\gamma+h'_\gamma\log^{O(1)}h'_\gamma)$ time, where $n'_\gamma$ and $h'_\gamma$ are the number
of vertices and chains in $P'_\gamma$ respectively.\footnote{
We could avoid Bar-Yehuda and Chazelle's algorithm and use some slower algorithm
running in $O(n'_\gamma + (h'_\gamma)^{O(1)})$ time, if we bootstrap one more time.
}
Note that $\sum_\gamma n'_\gamma = n'$ 
and $\sum_\gamma h'_\gamma = h'$.
Furthermore, $\max_\gamma h'_\gamma\le O(t\cdot s\log h)$ (since there are at most $O(t)$ sides in $\gamma$, and at most $O(t)$ start vertices of chains inside $\gamma$ due to those extra segments of zero length added to $R$).

\LIPICSv{\smallskip}
Afterwards, we can piece together $\TD(P'_\gamma)$ (after clipping to $\gamma$) over all regions $\gamma$,
to obtain $\TD(P)$ in $O(n')$ additional time.
\end{enumerate}

The total running time for Steps 0--1 is $O(n+h\log h)$.
The running time for Step 3 is $O(|R|\log|R|)=O((n/\sqrt{s} + h)\log n)$ time, as $|R|=O(n/\sqrt{s}+h)$.
The total running time for Steps 4--5 is at most
$O(n + ((|R|/\sqrt{t})\cdot s\log h\cdot \log n) + (h + (|R|/\sqrt{t})\cdot s\log h)\log^{O(1)}(t\cdot s\log h))$.

Setting $s=\log^2n$ and $t=\log^8n$ gives an overall time bound of
$O(n + h\log n + h (\log\log n)^{O(1)}) = O(n + h\log n)$, which is the same as $O(n + h\log h)$ (since the second term dominates only if $h\ge n/\log n$, but then $\log n=O(\log h)$).
\end{proof}

As a corollary, we can test whether a collection $P$ of $h$ polygonal chains has
any intersection in $O(n + h\log h)$ deterministic time: simply run the above algorithm to attempt
to compute $\TD(P)$, and run Chazelle's linear-time algorithm~\cite{Chazelle91} to check simplicity of the resulting (connected) diagram.


\section{When there are a sublinear number of intersections}\label{sec:intersect}

In this section, we consider the more general case when the polygonal chains may self-intersect or intersect
each other, with sublinear number of intersections.
Extra ideas are needed.  In particular, the interval-tree idea is no longer applicable, and edges may now cross
non-vertical sides of the boundary of a region.

In the following lemma, we begin with a special case where we have $b$ subcollections of chains with no intersections within each subcollection.

\begin{lemma}\label{lem:intersect}
Let $P$ be a collection of $h$ simple polygonal chains with a total of $n$ vertices.
Suppose $P$ is partitioned into $b$ subcollections $P_1,\ldots,P_b$ such that there are no
intersections within each subcollection~$P_i$.
Then we can compute $\TD(P)$ in $O(n + h\log h + b^{O(1)}X\log^{O(1)}n)$ deterministic time, where $X$ is the  number of intersections.
\end{lemma}
\begin{proof}
Let $\ell$ be the maximum number of vertices per chain in $P$.
Let $s$ and $t$ be parameters to be set later.
We describe an algorithm below to compute all intersections of $P$.
Steps 1--2 compute a subset $R$ satisfying a certain property.
Afterwards, Steps 3--5 use its trapezoidal decomposition $\TD(R)$ to divide the problem into smaller subproblems and then solve these subproblems.

\begin{enumerate}
\item
For each $i\in\{1,\ldots,b\}$:
\LIPICSv{\smallskip}
\begin{enumerate}
\item Apply Theorem~\ref{thm:disjoint} to compute $\TD(P_i)$
in $O(n_i+h_i\log h_i)$ time, where $n_i$ and $h_i$ are the number of vertices and chains in $P_i$
respectively.
\LIPICSv{\smallskip}
\item Compute an $s$-division of $\TD(P_i)$ in $O(n_i)$ time.
Preprocess the regions of the $s$-division  for point location~\cite{PreparataS85} in linear time.
Let $R_i$ be the subset of all $O(n_i/\sqrt{s})$ edges 
of $P_i$ that participate in the boundaries of regions in this $s$-division.
\end{enumerate}

\LIPICSv{\medskip}
\item
Let $R$ be the union of all the $R_i$'s, together with $O(h)$ extra line segments of zero length 
at the start vertices of the chains of $P$.  Note that $|R|=O(n/\sqrt{s} + h)$,
and $R$ satisfies the following property, similar to before:
\LIPICSv{\smallskip}
\begin{quote}
($\dag$)\ \ Given any vertical line segment $e$ not intersecting $R$, the segment $e$ can intersect at most $O(bs)$ edges of $P$; furthermore, these edges can be found in $O(b(\log n + s))$ time.
\end{quote}

\LIPICSv{\smallskip}
\item
Compute $\TD(R)$ in $O(|R|\log|R|+X)$ time, and
compute a $t$-division of $\TD(R)$ in $O(|R|+X)$ time.

\LIPICSv{\medskip}
\item
For each region $\gamma$ in the $t$-division and each vertical side $e$ of the boundary $\partial\gamma$,
compute all intersections of $e$ with $P$.
Since the interior of $e$ does not intersect $R$, there are at most $O(bs)$ such intersections and
they can be found in $O(b(\log n + s))$ time by property~($\dag$).
By summing over all $O((|R|+X)/\sqrt{t})$ choices of $e$, there are at most $O(((|R|+X)/\sqrt{t})\cdot bs)$ such intersections in total
and they can be found in $O(((|R|+X)/\sqrt{t})\cdot b(\log n+s))$ time.

\LIPICSv{\smallskip}
Cut the chains of $P$ at all these intersection points to obtain a new collection $P'$ of polygonal chains.
In order to cut,
we may need to sort the intersections along each edge of $P$, but this costs at most a logarithmic factor
per intersection, i.e., $O(((|R|+X)/\sqrt{t})\cdot bs\cdot \log n)$ in total.
The number of vertices in $P'$ is $n'\le n + O(((|R|+X)/\sqrt{t})\cdot bs)$,
and the number of chains in $P'$ is $h'\le h + O(((|R|+X)/\sqrt{t})\cdot bs)$.

\LIPICSv{\medskip}
\item
For each region $\gamma$ in the $t$-division,
let $P'_\gamma$ be the set of all chains of $P'$ with start vertices in $\gamma$,
let $Q_\gamma$ be the set of all non-vertical sides of $\partial\gamma$,
and apply 
Theorem~\ref{thm:disjoint} 
to test whether $P'_\gamma\cup Q_\gamma$ has any intersection
in $O(n'_\gamma + (h'_\gamma+|Q_\gamma|)\log n)$ time, where
$n'_\gamma$ and $h'_\gamma$ are the number of vertices and chains in $P'_\gamma$ respectively.
Note that $\sum_\gamma n'_\gamma = n'$ and $\sum_\gamma h'_\gamma=h'$ and $\sum_\gamma |Q_\gamma| = O((|R|+X)/\sqrt{t})$. 
If the answer is no, mark the region $\gamma$ \emph{good} and mark all the chains in $P'_\gamma$  \emph{good}; otherwise,
mark the region $\gamma$ \emph{bad} and mark all the chains in $P'_\gamma$ \emph{bad}.
If $\gamma$ is good, all chains in $P'_\gamma$ do not intersect each other and stay inside $\gamma$ (since
 chains in $P'$ do not intersect the vertical sides of $\partial\gamma$).
Thus, good chains do not intersect each other.
On the other hand, the number of bad regions is at most $O(X)$, since we can charge each bad region to a 
unique intersection (note that each chain is in just one $P'_\gamma$); so, the number of bad chains is at most $O(bstX)$, since each $P'_\gamma$ has at most $O(bs\cdot t)$ chains.

\LIPICSv{\medskip}
\item
For each bad chain $p'$ in $P'$, find all 
regions $\gamma$ in the $t$-division that are intersected by $p'$.
To do so, begin with the region $\gamma$ containing the start vertex of~$p'$.
Find the next line segment $e\in Q_\gamma$ intersected by $p'$, if exists, by brute force 
in $O(\ell t)$ time, since $p'$ has size at most $O(\ell)$ and $Q_\gamma$ has size at most $O(t)$.  Then reset $\gamma$ to the neighboring region defined by $e$, and repeat.
The total time is at most $O(\ell t\cdot bst X)$, since every region~$\gamma$ visited by $p'$ after the first can be charged to a unique intersection, and there are at most $O(bstX)$ bad chains~$p'$.

\LIPICSv{\medskip}
\item
For each bad chain $p'$ in $P'$ and each good region $\gamma$ intersected by $p'$ (all such regions have been identified from step~6), compute all intersections of $p'$ with
the good chains of $P'_\gamma$ by brute force in $(\ell\cdot \ell bst)$ time (since $P'_\gamma$ has at most 
$O(\ell\cdot bst)$ edges).  The total time is at most $O(\ell\cdot \ell bst\cdot bstX)$.

\LIPICSv{\medskip}
\item Lastly, compute all intersections among the bad chains, which have $n''=O(\ell\cdot bstX)$ edges, e.g., by the Bentley--Ottmann sweep~\cite{BentleyO79} in $O((n''+X)\log n'')$ time.
\end{enumerate}

The total running time for Step 1 is $O(n+h\log h)$.
The running time for Step 3 is $O(|R|\log|R|+X)=O((n/\sqrt{s} + h)\log n+X)$ time.
The total running time for Steps 4--5 is at most
$O(n + ((|R|+X)/\sqrt{t})\cdot bs\cdot \log n + h\log n)$.
The total running time for Steps 6--8 is $O((b\ell s t)^{O(1)}X)$.

Setting $s=\log^2n$ and $t=b^2\log^{6}n$ gives an overall time bound of
$O(n + h\log h + (b\ell)^{O(1)}X\log^{O(1)}n)$.

By initially subdividing the input chains into subchains of length $\log n$, we can lower $\ell$ to $\log n$ while
increasing $h$ to $h+O(n/\log n)$; the $O(h\log h)$ term would then be increased to $O(n+h\log h)$.

After computing all intersections, we can turn $P$ into a collection of at most $h$ disjoint 
simple chains with $O(n+X)$ vertices, by making tiny cuts to eliminate cycles and taking an 
Euler tour of each connected component
(to traverse the components, we may need to sort intersection points along each edge, but this costs at most $O(X\log n)$ time).
Afterwards, we can apply Theorem~\ref{thm:disjoint} to compute $\TD(P)$ in $O(n+X+h\log h)$ additional time.
\end{proof}

Finally, we combine the above lemma with an $n^{\Theta(\eps)}$-way divide-and-conquer to solve the general problem:

\begin{theorem}
Given a collection $P$ of $h$ (not necessarily simple) polygonal chains with a total of $n$ vertices,
we can compute $\TD(P)$ in $O(n + h\log h + Xn^\eps)$ deterministic time
for an arbitrarily small constant $\eps>0$, where $X$ is the number of intersections.
\end{theorem}
\begin{proof}
Divide $P$ into $b$ subcollections $P_1,\ldots,P_b$, each with $n/b$ vertices,
possibly by splitting at most $b$ of the input chains.
Suppose $P_i$ has $h_i$ chains and $X_i$ intersections.  Note that $\sum_i h_i\le h+b$ and $\sum_i X_i\le X$.
Recursively compute $\TD(P_i)$, whose output gives us the $X_i$ intersections of each $P_i$.  Subsequently, $P_i$ can be turned into a non-intersecting collection of 
at most $h_i$ simple chains with $O(n/b+X_i)$ vertices.
Now, apply Lemma~\ref{lem:intersect} to compute $\TD(P)$ in
$O(n+X + h\log h + b^{O(1)}X\log^{O(1)}n)$ time.

The total running time satisfies the recurrence
\[ T(n,h,X)\ \le \max_{\scriptsize\begin{array}{c}h_1,\ldots,h_b,\\X_1,\ldots,X_b:\\ \sum_i h_i\le h+b,\\ \sum_i X_i\le X\end{array}}\!\!
\left(\sum_{i=1}^b T(n/b,h_i,X_i) + O(n + h\log h + b^{O(1)}X\log^{O(1)}n)\right),
\] 
which solves to $T(n,h,X) = O((n + h\log h + b^{O(1)}X\log^{O(1)}n)\cdot \log_b n)$.
Setting $b=n^{\eps/c}$ for a sufficiently large constant $c$ yields the result.
\end{proof}

The main open question is whether the running time of the above theorem could
be improved to $O(n+h\log h+X)$.  A more modest goal is to find a deterministic
algorithm with running time $O(n + h\log h + X\log^{O(1)}n)$.

{\small

\PAPER{\bibliographystyle{alphaurl}}
\LIPICSv{\bibliographystyle{plainurl}}
\bibliography{polygon}

}

\appendix

\section{When randomization is allowed}\label{app:rand}

In this appendix, we describe a simpler randomized algorithm, which
can be viewed as a variant of an algorithm for $h=1$ by Eppstein, Goodrich, and Strash~\cite{EppsteinGS10}:

\begin{theorem}
Given a collection $P$ of $h$ (not necessarily simple) polygonal chains with a total of $n$ vertices,
we can compute $\TD(P)$ in $O(n + h\log h + X\log^{(c)}n)$ expected time
for an arbitrarily large constant $c$, where $X$ is the number of intersections.
\end{theorem}
\begin{proof}
Let $E$ be the edge set of $P$.  Let $s$ and $t$ be parameters to be set later. The algorithm works as follows:

\begin{enumerate}
\item Let $R$ be a random subset of $E$ of size $n/s$.

\LIPICSv{\medskip}
\item Compute $\TD(R)$ in $O(|R|\log |R| + X)$ time~\cite{ChazelleE92,ClarksonS89,Mulmuley90}, and compute
a $t$-division of $\TD(R)$.  Preprocess $\TD(R)$ for point location~\cite{PreparataS85} in $O(|R|)$ time.

\LIPICSv{\medskip}
\item For each trapezoid $\tau\in\TD(R)$, we generate the list $E_\tau$ of all edges of $E$ intersecting $\tau$
(called a ``conflict list'').
To do this, we begin with the trapezoid $\tau$ containing the start vertex of a chain, which can be
found in $O(\log n)$ time by point location.  Let $e$ be the first edge of the chain.
We then walk along $e$ in $\TD(R)$ to identify the next trapezoid intersected by $e$;
this requires $O(\DEG(\tau))$ time, where $\DEG(\tau)$ denotes the number of neighboring trapezoids of $\tau$.
We continue until we encounter the endpoint of $e$, then reset $e$ to the next edge, and repeat.
The total time is $O(h\log n)$ for the $h$ point location queries, plus
$O(\sum_{\tau\in\TD(R)} |E_\tau|\cdot \DEG(\tau))$ for the walks.

\LIPICSv{\smallskip}
Clarkson and Shor \cite[Lemma 4.3 and Corollary 4.4]{ClarksonS89} showed that 
\LIPICSv{\smallskip}
\PAPER{
\begin{enumerate}
\item[(i)] $\sum_{\tau\in\TD(R)} |E_\tau|\cdot \DEG(\tau)$ has expected value $O( (n/s + X/s^2)\cdot s) = O(n+X/s)$;
\item[(ii)] $\max_{\tau\in\TD(R)} |E_\tau| = O(s\log n)$ with probability at least (say) $1-O(1/n^2)$.
\end{enumerate}
}\LIPICSv{
\begin{quote}
\begin{enumerate}
\item[(i)] $\sum_{\tau\in\TD(R)} |E_\tau|\cdot \DEG(\tau)$ has expected value $O( (n/s + X/s^2)\cdot s) = O(n+X/s)$;
\item[(ii)] $\max_{\tau\in\TD(R)} |E_\tau| = O(s\log n)$ with probability at least (say) $1-O(1/n^2)$.
\end{enumerate}
\end{quote}
}
\LIPICSv{\smallskip}
If (ii) fails, we can afford to switch to a naive $O(n^2)$-time algorithm.
So, from now on, assume that (ii) holds.

\LIPICSv{\medskip}
\item For each region $\gamma$ in the $t$-division and each (vertical or non-vertical) side
$e$ of the boundary $\partial\gamma$, find the intersections of $e$ with $P$.
There are $O(s\log n)$ such intersections by (ii) and
they can be generated in $O(s\log n)$ time (from the output of Step~3).
By summing over all $O(|R|/\sqrt{t})$ choices of $e$, there are $O((|R|/\sqrt{t})\cdot s\log n)$ such intersections in total and
they can be generated in $O((|R|/\sqrt{t})\cdot s\log n)$ time.

\LIPICSv{\smallskip}
Cut the chains of $P$ at all these intersection points to obtain a new collection $P'$ of polygonal chains.
Note that the intersections are generated in sorted order along the chains from Step~3.
The number of vertices in $P'$ is $n'\le n + O((|R|/\sqrt{t})\cdot s\log n)$,
and the number of chains in $P'$ is $h'\le h + O((|R|/\sqrt{t})\cdot s\log n)$.

\LIPICSv{\medskip}
\item 
For each region $\gamma$ in the $t$-division,
let $P'_\gamma$ be the chains of $P'$ inside $\gamma$,
and compute $\TD(P'_\gamma)$ recursively.
Let $n'_\gamma$ and $h'_\gamma$ be the number
of vertices and chains in $P'_\gamma$ respectively.
Note that $\sum_\gamma n'_\gamma = n'$ 
and $\sum_\gamma h'_\gamma = h'$.
Furthermore, $\max_\gamma n'_\gamma\le O(t\cdot s\log n)$ by (ii).

\LIPICSv{\smallskip}
Afterwards, we can piece together $\TD(P'_\gamma)$ (after clipping to $\gamma$) over all regions $\gamma$,
to obtain $\TD(P)$ in $O(n')$ additional time.
\end{enumerate}

Set $s=\log n$ and $t=\log^6 n$.
The total expected running time satisfies the recurrence
\[ T(n,h,X)\ \le \max_{\scriptsize\begin{array}{c}\sum_\gamma n'_\gamma\le n + O(n/\log n)\\
\sum_\gamma  h'_\gamma\le h+O(n/\log n)\\  \sum_\gamma X_\gamma \le X\\ \max_\gamma n'_\gamma\le O(\log^8n)\end{array}} 
\left(\sum_\gamma T(n'_\gamma,h'_\gamma,X_\gamma) + O(n + h\log n + X)\right).
\] 

For the base case, we have $T(n,h,0)=O(n+h\log^{O(1)}n)$ by Bar-Yehuda and Chazelle's 
 algorithm~\cite{Bar-YehudaC94}, and $T(n,h,X)=O(n^2)$ by a naive algorithm for $X>0$.
Applying the recursion for $c+1$ levels before switching to the base case yields a
total time bound of $O(n + h\log n  + \sum_{j=1}^{c+1} (h+n/\log^{(j)}n) \log^{O(1)}(\log^{(j)}n) + X((\log^{(c+1)}n)^8)^2)
\le O(n + h\log h + X\log^{(c)}n)$.
\end{proof}

\end{document}